\documentclass{article}

\usepackage{amssymb}
\usepackage{amsmath}
\usepackage{scalerel}

\newtheorem{definition}{Definition}
\newtheorem{theorem}{Theorem}
\newtheorem{lemma}{Lemma}

\newenvironment{proof}{\noindent{\sf Proof.}}{\hfill $\boxtimes$\linebreak}
\newcommand{\qed}{\hfill $\boxtimes$\linebreak}
\renewcommand{\phi}{\varphi}
\renewcommand{\epsilon}{\varepsilon}

\title{Equilibria Interchangeability in Cellular Games}

\author{Pavel Naumov \and Margaret Protzman}

\date{}

\begin{document}

\maketitle

\begin{center}
\vspace{-3mm}
Department of Mathematics and Computer Science\\
McDaniel College, Westminster, Maryland, USA\\ \vspace{3mm}
{\sf \{pnaumov,mmp001\}@mcdaniel.edu}
\vspace{3mm}
\end{center}

\begin{abstract}
The notion of interchangeability has been introduced by John Nash in one of his original papers on equilibria. This paper studies properties of Nash equilibria interchangeability in cellular games that model behavior of infinite chain of homogeneous economic agents. The paper shows that there are games in which strategy of any given player is interchangeable with strategies of players in an arbitrary large neighborhood of the given player, but is not interchangeable with the strategy of a remote player outside of the neighborhood. The main technical result is a sound and complete logical system describing universal properties of  interchangeability common to all cellular games.
\end{abstract}

\section{Introduction}

\paragraph{Cellular Games.}

An one-dimensional cellular automaton is an infinite row of cells that transition from one state to another under certain rules. The rules are assumed to be identical for all cells. Usually, rules are chosen in such a way that the next state of each cell is determined by the current states of the cell itself and its two neighboring cells.

Harjes and Naumov~\cite{hn13lori} introduced an object similar to cellular automaton that they called {\em cellular game}. They proposed to view each cell as a player, whose pay-off function depends on the strategy of the cell itself and the strategies of its two neighbors. The cellular games are {\em homogeneous} in the sense that all players of a given game have the same pay-off function. Such games can model rational behavior of linearly-spaced homogeneous agents. Linearly-spaced economies have been studied by economists before~\cite{s71jet}.

Consider an example of a cellular game that we call $G_1$. In this game each player has only three strategies. We identify these strategies with congruence classes $[0]$, $[1]$, and $[2]$ of $\mathbb Z_3$. Each player is rewarded for either matching the strategy of her left neighbor or choosing strategy one more (in $\mathbb Z_3$) than the strategy of the left neighbor. An example of a Nash equilibrium in this game is shown on Figure~\ref{introG1}.
\begin{figure}[ht]
\begin{center}
\vspace{3mm}
\scalebox{.7}{\includegraphics{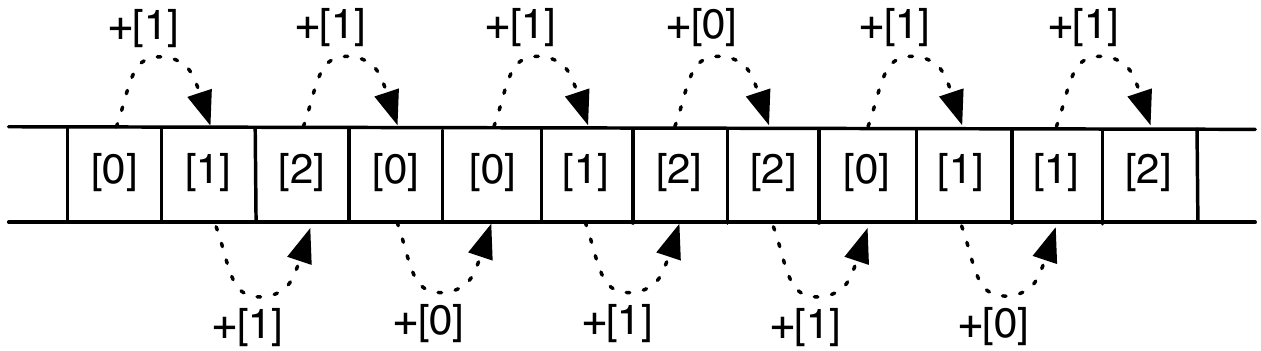}}
\vspace{0mm}
\caption{A Nash equilibrium of game $G_1$.}\label{introG1}
\vspace{0cm}
\end{center}
\vspace{0cm}
\end{figure}

\paragraph{Interchangeability.}

The notion of interchangeability goes back to one of Nash's original papers~\cite{n51am} on equilibria in strategic games. Interchangeability is easiest to define in a two-player game: players in such a game are interchangeable if for any two equilibria $\langle a_1,b_1\rangle$ and $\langle a_2,b_2\rangle$, strategy profiles $\langle a_1,b_2\rangle$ and $\langle a_2,b_1\rangle$ are also equilibria. Players in any two-player {\em zero-sum} game are interchangeable~\cite{n51am}.

Consider now a multiplayer game with set of players $P$. We say that players $p\in P$ and $q\in P$ are interchangeable if for any two equilibria $\langle e'_i\rangle_{i\in P}$ and $\langle e''_i\rangle_{i\in P}$ of the game, there is equilibrium $\langle e_i\rangle_{i\in P}$ of the same game such that $e_p=e'_p$ and $e_q=e''_q$. We denote this by $p\parallel q$.  For example, it is easy to see that for the game described in the previous section, players $p$ and $q$ are interchangeable if there are not adjacent. In other words, $p\parallel q$ if and only if $|p-q|>1$. 
This is the relation whose properties in cellular games we study in this paper.

We now consider another game, that we call $G_2$. Each player in  game $G_2$ can either pick a strategy from $\mathbb Z_3$ or switch to playing matching pennies game with both of her neighbors. In the latter case, the strategy is a pair $(y_1,y_2)$, where $y_1,y_2\in \{head,tail\}$. Value $y_1$ is the strategy in the matching pennies game against the left neighbor and value $y_2$ is the strategy against the right neighbor. 

If the left and the right neighbors of a player choose, respectively, elements $x$ and $z$ from set $\mathbb Z_3$ such that $z-x\in \{[0],[1]\}$, then the player is not paid no matter what her strategy is.  Otherwise, player is rewarded to start  matching pennies games with both neighbors. If two adjacent players both play matching pennies game, then player on the right is rewarded to {\em match} the penny of the player on the left and player on the left is rewarded to {\em mismatch} the penny of the player on the right. An example of a Nash equilibria in such game is shown on Figure~\ref{introG2}.
\begin{figure}[ht]
\begin{center}
\vspace{3mm}
\scalebox{.7}{\includegraphics{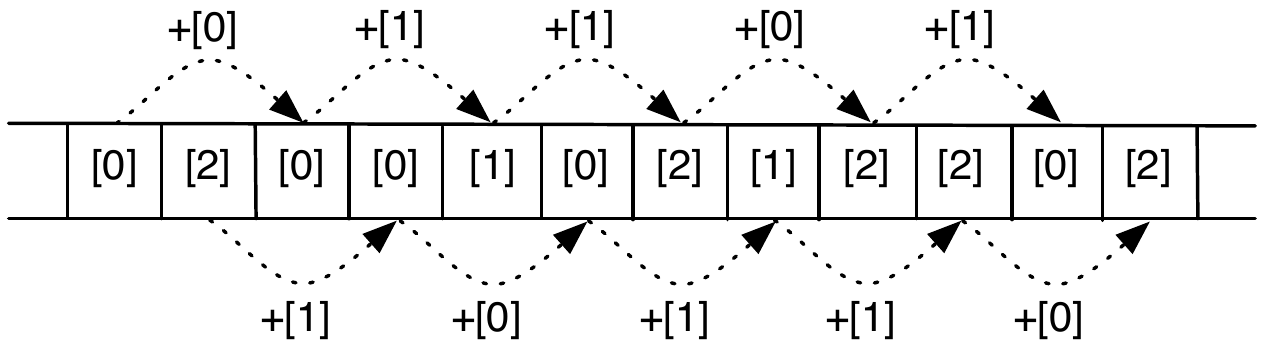}}
\vspace{0mm}
\caption{A Nash equilibrium of game $G_2$.}\label{introG2}
\vspace{0cm}
\end{center}
\vspace{0cm}
\end{figure}
The set of all Nash equilibria of this game consists of all strategy profiles in which each player chooses an element of $\mathbb Z_3$ in such a way that for each player $p$ player $p+2$ never chooses strategy that is two-more (in $\mathbb Z_3$) than the strategy of player $p$. An interesting property of this game  (see Theorem~\ref{g2 theorem}) is that $p\parallel q$ if and only if $|p-q|\neq 2$ and $p\neq q$. Thus, any two adjacent players are interchangeable, but players that are two-apart are not interchangeable. Note that we have achieved this by using each player to synchronize  strategies of her two neighbors. The ability of a player to do this significantly relies on the fact that pay-off function of each player is computed based on the choice of strategies by the player herself and her two adjacent neighbors. A player can not synchronize in the same way strategies of the players that are 2-players away to the left and 2-players away to the right. Thus, it would be natural to assume that it is impossible to construct a cellular game in which players, say 1000-players away, are non-interchangeable, but players that are closer are interchangeable.  If this hypothesis is true, then the following property is true for all cellular games:
\begin{equation}\label{hypothesis}
p \parallel (p+1) \wedge p\parallel (p+2) \wedge \dots \wedge p \parallel (p+999) \to p\parallel (p + 1000).
\end{equation}

The main surprising result of this paper is that such game does exist. Namely, we prove that for any $n\ge 1$ there is a cellular game $G_n$ in which any two players $p$ and $q$ are interchangeable if and only if $|p-q|=n$. The construction of such game for $n>2$ is non-trivial. Strategies of players in our game are special $(n-1)\times 2$ matrices of elements from $\mathbb Z_{n+1}$. 

Another way to state our result is to say that statement~(\ref{hypothesis}) is not a universal property of cellular games. Naturally, one can ask what statements are universal properties of all cellular games. We answer this question by giving a sound and complete axiomatization of such properties consisting of just the following three axioms:
\begin{enumerate}
\item Reflexivity: $a\parallel a\rightarrow a\parallel b$,
\item Homogeneity: $a\parallel b \rightarrow (a+c)\parallel (b+c)$,
\item Symmetry: $a\parallel b \rightarrow b\parallel a$.
\end{enumerate}
The proof of completeness takes multiple instances of the discussed above cellular game $G_n$ and combines them into a single cellular game needed to finish the proof.

The interchangeability relation between players of multi-player game could be further generalized to a relation between two sets of players. Properties of this relation are completely axiomatizable~\cite{nn11lori} by  Geiger, Paz, and Pearl axioms originally proposed to describe properties of independence in the probability theory~\cite{gpp91}. The same axioms also describe properties of Sutherland's~\cite{s86} nondeducibility relation in information flow theory~\cite{mn10} and of a non-interference relation in concurrency theory~\cite{mns11csl}. Naumov and Simonelli~\cite{ns12loft} described interchangeability properties between two sets of players in zero-sum games.

\paragraph{Functional Dependence.} Our work is closely related to paper by Harjes and Naumov~\cite{hn13lori} on functional dependence in cellular games.  Strategy of player $p$ functionally determines strategy of player $q$ in a cellular game if any two Nash equilibria of the game that agree on player $p$ also agree on player $q$. We denote this by $p\rhd q$. The functional dependence relation between players can not be expressed through interchangeability and vice versa. Harjes and Naumov gave complete axiomatization of functional dependence relation for cellular games with finite set of strategies:
\begin{enumerate}
\item Reflexivity: $a\rhd a$,
\item Transitivity: $a\rhd b \rightarrow (b\rhd c\rightarrow a\rhd c)$,
\item Homogeneity: $a\rhd b \rightarrow (a+c)\rhd (b+c)$,
\item Symmetry: $a\rhd b \rightarrow b\rhd a$.
\end{enumerate}
In spite of certain similarity between these axioms and our axioms for interchangeability, the proofs of completeness are very different. The completeness proof techniques used by Harjes and Naumov is based on properties of Fibonacci numbers and, to the best of our knowledge, can not be adopted to our setting. Similarly, the $(n-1)\times 2$-matrix based game $G_n$ that we use in the current paper can not be used to prove the results obtained in~\cite{hn13lori}.

The paper is structured as following. In Section~\ref{syntax and semantics}, we give the formal definition of a cellular game and introduce formal syntax and semantics of our theory. In Section~\ref{axioms}, we list the axioms of out logical systems and review some related notations. In Section~\ref{soundness}, we prove soundness of this logical system. The rest of the paper is dedicated to the proof of completeness. In  Section~\ref{g0}, Section~\ref{g1}, and Section~\ref{g2}, we define special cases of the game $G_n$ for $n=0,1,2$ and prove their key properties. Games $G_1$ and $G_2$ has already been informally discussed above. In Section~\ref{gn} we give general definition of $G_n$ for $n\ge 3$ and prove its properties. We combine results about games $G_n$ for all $n\ge 0$ in Section~\ref{combine}. In Section~\ref{ginf}, we introduce a very simple game $G_\infty$ and prove its properties.  In Section~\ref{product section}, we define a product operation on cellular games that can be used to combine several cellular games into one.  In Section~\ref{final steps}, we use the product of multiple games $G_n$ to finish the proof of completeness. Section~\ref{conclusion} concludes.

\section{Syntax and Semantics}\label{syntax and semantics}

In this section we formally define cellular games, Nash equilibrium, and introduce the formal syntax and the formal semantics of our logical system. The definition of interchangeability predicate $a\parallel b$ is a part of the formal semantics specification in Definition~\ref{true} below.

\begin{definition}\label{}
Let $\Phi$ be the minimal set of formulas that satisfies the following conditions:
\begin{enumerate}
\item $\bot\in\Phi$,
\item $a\parallel b \in \Phi$ for each integer $a,b\in \mathbb Z$,
\item if $\varphi\in\Phi$ and $\psi\in\Phi$, then $\varphi\rightarrow\psi\in\Phi$.
\end{enumerate}
\end{definition}

\begin{definition}\label{}
Cellular game is a pair $(S,u)$, where 
\begin{enumerate}
\item $S$ is any set of ``strategies",
\item $u$ is a ``pay-off" function from $S^3$ to the set of real numbers $\mathbb R$.
\end{enumerate}
\end{definition}
The domain of the function $u$ in the above definition is $S^3$ because the pay-off of each player is determined by her own strategy and the strategies of her two neighbors. By a strategy profile of a cellular game $(S,u)$ we mean any tuple $\langle s_i\rangle _{i\in \mathbb Z}$ such that $s_i\in S$ for each $i\in \mathbb Z$. 
\begin{definition}
A Nash equilibrium of a game $(S,u)$ is any strategy profile $\langle e_i\rangle _{i\in \mathbb Z}$ such that 
$u(e_{i-1},s,e_{i+1})\le u(e_{i-1},e_i,e_{i+1})$,
for each $i\in\mathbb Z$ and each $s\in S$.
\end{definition}
By $NE(G)$ we denote the set of all Nash equilibria of a cellular game $G$.
\begin{lemma}\label{ne-shift-lemma}
For each $k\in \mathbb Z$, if $\langle e_i\rangle _{i\in \mathbb Z}$ is a Nash equilibrium of a cellular game, then $\langle e_{i+k}\rangle _{i\in \mathbb Z}$ is a Nash equilibrium of the same game. \hfill \qed
\end{lemma}

\begin{definition}\label{true}
For any formula $\varphi\in\Phi$ and any cellular game $G$, relation $G\vDash\varphi$ is defined recursively as follows:
\begin{enumerate}
\item $G\nvDash\bot$,
\item $G\vDash a\parallel b$ if and only if for each $\langle e'_i\rangle _{i\in \mathbb Z}\in NE(G)$ and each $\langle e''_i\rangle _{i\in \mathbb Z}\in NE(G)$, there is $\langle e_i\rangle _{i\in \mathbb Z}\in NE(G)$ such that $e_a=e'_a$ and $e_b=e''_b$,
\item $G\vDash \psi\rightarrow \chi$ if and only if $G\nvDash \psi$ or $G\vDash \chi$.
\end{enumerate}
\end{definition}

\section{Axioms}\label{axioms}

Our logical system, in addition to propositional tautologies in the language $\Phi$ and the Modus Ponens inference rule, contains the following axioms:

\begin{enumerate}
\item Reflexivity: $a\parallel a\rightarrow a\parallel b$,
\item Homogeneity: $a\parallel b \rightarrow (a+c)\parallel (b+c)$,
\item Symmetry: $a\parallel b \rightarrow b\parallel a$.
\end{enumerate}

We write $\vdash\varphi$ if formula $\phi$ is provable in our logical system. The next lemma gives an example of a proof in our logical system. This lemma will later be used in the proof of the completeness theorem. 
\begin{lemma}\label{abs value lemma}
If $|a-b|=|c-d|$, then $\vdash a\parallel b \to c \parallel d$.
\end{lemma}
\begin{proof}
Due to Symmetry axiom, without loss of generality we can assume that $a>b$ and $c>d$. Thus, assumption $|a-b|=|c-d|$ implies that $a-b=c-d$. Hence, $c - a = d -b$. Then, by Homogeneity axiom, $$\vdash a\parallel b \to (a+ (c-a))\parallel (b+(d-b)).$$ In other words, $\vdash a\parallel b \to c \parallel d$.
\end{proof}

\section{Soundness}\label{soundness}

Soundness of propositional tautologies and Modus Ponens inference rules is straightforward. We prove soundness of each of the remaining axioms of our logical system as a separate lemma.

\begin{lemma}[reflexivity]\label{reflexivity sound}
If $G\vDash a\parallel a$, then $G\vDash a\parallel b$ for each $a,b\in\mathbb Z$.
\end{lemma}
\begin{proof}
Let $\langle e'_i\rangle _{i\in \mathbb Z}\in NE(G)$ and $\langle e''_i\rangle _{i\in \mathbb Z}\in NE(G)$, we need to show that there exists $\langle e_i\rangle _{i\in \mathbb Z}\in NE(G)$ such that $e_a=e'_a$ and $e_b=e''_b$. Indeed, by assumption $G\vDash a\parallel a$, there exists $\langle e'''_i\rangle _{i\in \mathbb Z}\in NE(G)$ such that $e'''_a=e'_a$ and $e'''_a=e''_a$. Thus, $e'_a=e''_a$. Take $\langle e_i\rangle _{i\in \mathbb Z}\in NE(G)$ to be $\langle e''_i\rangle _{i\in \mathbb Z}\in NE(G)$. Then $e_a=e''_a=e'_a$ and $e_b=e''_b$.
\end{proof}

\begin{lemma}[homogeneity]\label{homogeneity sound}
If $G\vDash a\parallel b$, then $G\vDash (a+c)\parallel (b+c)$, for each $a,b,c\in\mathbb Z$.
\end{lemma}
\begin{proof}
Let $\langle e'_i\rangle _{i\in \mathbb Z}\in NE(G)$ and $\langle e''_i\rangle _{i\in \mathbb Z}\in NE(G)$, we need to show that there exists $\langle e_i\rangle _{i\in \mathbb Z}\in NE(G)$ such that $e_{a+c}=e'_{a+c}$ and $e_{b+c}=e''_{b+c}$. By Lemma~\ref{ne-shift-lemma}, $\langle e'_{i+c}\rangle _{i\in \mathbb Z}\in NE(G)$ and $\langle e''_{i+c}\rangle _{i\in \mathbb Z}\in NE(G)$. Then, by assumption $G\vDash a\parallel b$, there exists $\langle e'''_i\rangle _{i\in \mathbb Z}\in NE(G)$ such that $e'''_a=e'_{a+c}$ and $e'''_b=e''_{b+c}$. Lemma~\ref{ne-shift-lemma} implies that $\langle e'''_{i-c}\rangle _{i\in \mathbb Z}\in NE(G)$. Take $\langle e_i\rangle _{i\in \mathbb Z}\in NE(G)$ to be $\langle e'''_{i-c}\rangle _{i\in \mathbb Z}$. Then, $e_{a+c}=e'''_{(a+c)-c}=e'''_a=e'_{a+c}$ and $e_{b+c}=e'''_{(b+c)-c}=e'''_b=e''_{b+c}$.
\end{proof}

\begin{lemma}[symmetry]\label{symmetry sound}
If $G\vDash a\parallel b$, then $G\vDash b\parallel a$ for each $a,b\in\mathbb Z$.
\end{lemma}
\begin{proof}
Let $\langle e'_i\rangle _{i\in \mathbb Z}\in NE(G)$ and $\langle e''_i\rangle _{i\in \mathbb Z}\in NE(G)$. We need to show that there is $\langle e_i\rangle _{i\in \mathbb Z}\in NE(G)$ such that $e_b=e'_b$ and $e_a=e''_a$. Indeed, by assumption $G\vDash a\parallel b$, there exists $\langle e_i\rangle _{i\in \mathbb Z}\in NE(G)$ such that $e_a=e''_a$ and $e_b=e'_b$. 
\end{proof}

\section{Completeness}

In this section we prove completeness of our logical system by showing that for each formula $\phi$ such that $\nvdash\phi$ there exists a cellular game $G$ such that $G\nvDash\phi$. The game $G$ will be constructed as a composition of multiple cellular ``mini" games $G_n$. Throughout this paper, by $[k]_n$ we mean the equivalence class of $k$ modulo $n$. In other words, $[k]_n\in\mathbb Z_n$. We sometimes omit subscript $n$ in the expression $[k]_n$ if the value of the subscript is clear from the context. While proving properties of the game $G_n$, we will find useful the following technical lemma:

\begin{lemma}\label{little lemma}
For any $n\ge 1$, any $u,v\in\mathbb Z_n$, and any $k\ge n$, there is a sequence of classes $z_1,\dots,z_k\in \mathbb Z_n$ such that
\begin{enumerate}
\item $z_1=u$,
\item $z_k=v$,
\item $z_{i+1}-z_i\in\{[0]_n,[1]_n\}$ for each $i<k$.
\end{enumerate}
\qed
\end{lemma}

\subsection{Game $G_0$}\label{g0}
We start with a very simple game $G_0$.
\begin{definition}\label{}
Let $G_0$ be pair $(\mathbb Z_2, 0)$, where pay-off function is constant $0$.
\end{definition}
\begin{lemma}\label{g0ne}
The set of all Nash equilibria of the game $G_0$ is set of all possible strategy profiles of this game.\qed
\end{lemma}
\begin{theorem}\label{g0 theorem}
$G_0\vDash a\parallel b$ if and only if $a\neq b$.
\end{theorem}
\begin{proof}
$(\Rightarrow):$ Suppose that $G_0\vDash a\parallel b$ and $a=b$. Consider strategy profiles $\langle e'_k\rangle_{k\in \mathbb Z}$ and $\langle e''_k\rangle_{k\in \mathbb Z}$ such that $e'_k=[0]$ and $e''_k=[1]$ for each $k\in \mathbb Z$. By Lemma~\ref{g0ne}, strategy profiles $\langle e'_k\rangle_{k\in \mathbb Z}$ and $\langle e''_k\rangle_{k\in \mathbb Z}$ are Nash equilibria of the game $G_0$. Thus, by the assumption $G_0\vDash a\parallel b$, there must exist Nash equilibrium $\langle e_k\rangle_{k\in \mathbb Z}$ such that $e_a=e'_a$ and $e_b=e'_b$. Recall that $a=b$. Thus, $[0]_2=e'_a=e_a=e_b=e''_b=[1]_2$, which is a contradiction. 
  
\noindent $(\Leftarrow):$ Assume that $a\neq b$ and consider any two Nash equilibria $\langle e'_k\rangle_{k\in \mathbb Z}$ and $\langle e''_k\rangle_{k\in \mathbb Z}$ of the game $G_0$. We need to show that there is Nash equilibrium $\langle e_k\rangle_{k\in \mathbb Z}$ such that $e_a=e'_a$ and $e_b=e''_b$. Indeed, consider strategy profile $\langle e_k\rangle_{k\in \mathbb Z}$ such that
$$
e_k=
\begin{cases}
e'_a 	& \mbox{ if $k=a$},\\
e''_b 	& \mbox{ if $k=b$},\\
[0]_2	& \mbox{ otherwise.}
\end{cases}
$$
By Lemma~\ref{g0ne}, strategy profile $\langle e_k\rangle_{k\in \mathbb Z}$ is a Nash equilibrium of the game $G_0$.
 \end{proof}

\subsection{Game $G_1$}\label{g1}

Let us now recall from the introduction the definition of game $G_n$ for $n=1$ and prove its important property. Each player in this game has only three strategies. We identify these strategies with congruence classes in $\mathbb Z_3$. Each player is rewarded if she either matches the strategy of her left neighbor or chooses the strategy one more (in $\mathbb Z_3$) than the strategy of the left neighbor. This is formally specified by the definition below.

\begin{definition}\label{}
Let game $G_1$ be pair $(\mathbb Z_3, u)$, where 
$$
u(x,y,z)=
\begin{cases}
1 	& \mbox{ if $y \neq x +[2]_3$},\\
0	& \mbox{ otherwise.}
\end{cases}
$$
\end{definition}
An example of a Nash equilibrium of game $G_1$ is depicted in Figure~\ref{introG1} in the introduction.

\begin{lemma}\label{g1ne}
Strategy profile $\langle e_k\rangle_{k\in \mathbb Z}$ is a Nash equilibrium of the game $G_1$ if and only if $e_k-e_{k-1}\in \{[0]_3,[1]_3\}$ for each $k\in \mathbb Z$. \qed
\end{lemma}

\begin{theorem}\label{g1 theorem}
$G_1\vDash a\parallel b$ if and only if $|a-b| > 1$.
\end{theorem}
\begin{proof} 
Without loss of generality, we can assume that $a\le b$.

\noindent$(\Rightarrow)$ 
First, suppose that $a=b$. Consider strategy profiles  $e'=\langle e'_k\rangle_{k\in \mathbb Z}$ and $e''=\langle e''_k\rangle_{k\in \mathbb Z}$ of the game $G_1$ such that $e'_k=[0]_3$ and $e''_k=[1]_3$ for each $k\in \mathbb Z$. By Lemma~\ref{g1ne}, $e',e''\in NE(G_1)$. Assume that $G_1\vDash a\parallel b$, then there must exist $e=\langle e_k\rangle_{k\in \mathbb Z}\in NE(G_1)$ such that $e_a=e'_a$ and $e_b=e''_b$. Thus, $[0]_3=e'_a=e_a=e_b=e''_b=[1]_3$, due to the assumption $a=b$. Therefore, $[0]_3=[1]_3$, which is a contradiction.

Next, assume that $b=a+1$. Consider strategy profile $e'=\langle e'_k\rangle_{k\in \mathbb Z}$ such that $e'_k=[0]_3$ for each $k\in \mathbb Z$ and  strategy profile $e''=\langle e''_k\rangle_{k\in \mathbb Z}$ such that $e''_k=[2]_3$ for each $k\in \mathbb Z$. By Lemma~\ref{g1ne}, $e',e''\in NE(G_1)$. At the same time, due to the same Lemma~\ref{g1ne}, there can not be $e=\langle e_k\rangle_{k\in \mathbb Z}\in NE(G_1)$ such that $e_a=[0]_3$ and $e_b=e_{a+1}=[2]_3$. Therefore, $G_1\nvDash a\parallel b$.

\noindent $(\Leftarrow)$ Assume that $|a-b|>1$. Thus, $a+1<b$ due to the assumption $a\le b$. Consider any two equilibria $e'=\langle e'_k\rangle_{k\in \mathbb Z}$ and $e''=\langle e'_k\rangle_{k\in \mathbb Z}$ of game $G_1$. We will show that there is $e=\langle e_k\rangle_{k\in \mathbb Z}\in NE(G_1)$ such that $e_a=e'_a$ and $e_b=e''_b$. Indeed, since $a+1<b$, by Lemma~\ref{little lemma}, there must exist sequence of congruence classes $x_a,x_{a+1},x_{a+2},\dots,x_b$  in $\mathbb Z_3$ such that  $x_{k} - x_{k-1} \in\{[0]_3,[1]_3\}$ for each $a<k\le b$. Define strategy profile $e=\langle e_k\rangle_{k\in \mathbb Z}$ as
$$
e_k=
\begin{cases}
x_a 	& \mbox{ if $k< a$},\\
x_k 	& \mbox{ if $a\le k\le b$},\\
x_b	& \mbox{ if $b < k$}.
\end{cases}
$$
By Lemma~\ref{g1ne}, $e\in NE(G_1)$.
\end{proof}

\subsection{Game $G_2$}\label{g2}

We now recall definition of game $G_n$ for $n=2$ from the introduction. Each player in this game can either pick a strategy from $\mathbb Z_3$ or switch to playing matching pennies game with both of her neighbors. In the latter case, the strategy is a pair $(y_1,y_2)$, where $y_1,y_2\in \{head,tail\}$. Value $y_1$ is the strategy in the matching pennies game against the left neighbor and value $y_2$ is the strategy against the right neighbor. 

If the left and the right neighbors of a player choose, respectively, elements $x$ and $z$ from set $\mathbb Z_3$ such that $z-x\in \{[0]_3,[1]_3\}$, then the player is not paid no matter what her strategy is.  Otherwise, player is rewarded to start  matching pennies games with both neighbors. If two adjacent players both play matching pennies game, then player on the right is rewarded to {\em match} the penny of the player on the left and player on the left is rewarded to {\em mismatch} the penny of the player on the right. We formally capture the above description of the game $G_2$ in the following definition.

\begin{definition}\label{}
Let game $G_2$ be pair $(S, u)$, where 
\begin{enumerate}
	
	\item $S=\mathbb Z_3\cup \{head,tail\}^2$. In other words, strategy of each player 
		in this game could be either a congruence class from $\mathbb Z_3$ or a pair $(y_1,y_2)$ 
		such that each of $u$ and $v$ is either ``head" or ``tail".
		
	\item pay-off function $u(x,y,z)=u_1(x,y,z)+u_2(x,y)+u_3(y,z)$ is the sum of three 
		separate pay-offs specified below:
		
	\begin{enumerate}
	
		\item if either at least one of $x$ and $z$ is not in $\mathbb Z_3$ or if they are both in $\mathbb Z_3$
			and $x+[2]_3=z$, then pay-off $u_1(x,y,z)$ rewards player $y$ 
			not to be an element of $\mathbb Z_3$:
			$$u_1(x,y,z)=
				\begin{cases}
					1 	& \mbox{ if $y\in \{head,tail\}^2$},\\
					0 	& \mbox{ otherwise}.
				\end{cases}
			$$
			in all other cases $u_1(x,y,z)$ is equal to zero.
		
		\item if both $x=(x_1,x_2)$ and $y=(y_1,y_2)$ are in $\{head,tail\}^2$, then pay-off $u_2(x,y)$ rewards
			player $y$ if $x_2 = y_1$:
			$$u_2((x_1,x_2),(y_1,y_2))=
				\begin{cases}
					1 	& \mbox{ if $x_2= y_1$},\\
					0 	& \mbox{ otherwise}.
				\end{cases}
			$$
			in all other cases $u_2(x,y)$ is equal to zero.
			
		\item if both $y=(y_1,y_2)$ and $z=(z_1,z_2)$ are in $\{head,tail\}^2$, then pay-off $u_3(y,z)$ rewards
			player $y$ if $y_2\neq z_1$:
			$$u_2((y_1,y_2),(z_1,z_2))=
				\begin{cases}
					1 	& \mbox{ if $y_2\neq z_1$},\\
					0 	& \mbox{ otherwise}.
				\end{cases}
			$$
			in all other cases $u_3(y,z)$ is equal to zero.		 
	\end{enumerate}
\end{enumerate}
\end{definition}
An example of a Nash equilibrium of the game $G_2$ has been given in the introduction in Figure~\ref{introG2}.

\begin{lemma}\label{g2ne}
Strategy profile $\langle e_k\rangle_{k\in \mathbb Z}$ is a Nash equilibrium of game $G_2$ if and only if the following two conditions are satisfied:
\begin{enumerate}
\item $e_k\in \mathbb Z_3$ for each $k\in \mathbb Z$,
\item $e_{k+2}-e_k\in \{[0]_3,[1]_3\}$ for each $k\in \mathbb Z$.
\end{enumerate} \qed
\end{lemma}

\begin{theorem}\label{g2 theorem}
$G_2\vDash a\parallel b$ if and only if either $|a-b|=1$ or $|a-b| > 2$.
\end{theorem}
\begin{proof}
Without loss of generality, suppose that $a\le b$.

\noindent$(\Rightarrow)$ First, assume that $G_2\vDash a\parallel b$ and $a=b$. Consider strategy profiles $e'=\langle e'_k\rangle_{k\in \mathbb Z}$ and $e'=\langle e''_k\rangle_{k\in \mathbb Z}$ such that $e'_k = [0]_3$ and $e''_k = [1]_3$ for each $k\in \mathbb Z$. Note that $e',e''\in NE(G_2)$ by Lemma~\ref{g2ne}. Thus, by the assumption $G_2\vDash a\parallel b$, there must exist $e=\langle e_k\rangle_{k\in\mathbb Z}$ such that $e_a=e'_a=[0]_3$ and $e_b=e''_b=[1]_3$. Hence, because $a=b$, we have $[0]_3=e_a=e_b=[1]_3$, which is a contradiction.

Next, suppose that $G_2\vDash a\parallel b$ and $b = a + 2$.  Consider strategy profiles $e'=\langle e'_k\rangle_{k\in \mathbb Z}$ and $e'=\langle e''_k\rangle_{k\in \mathbb Z}$ such that $e'_k = [0]_3$ and $e''_k = [2]_3$ for each $k\in \mathbb Z$. Note that $e',e''\in NE(G_2)$ by Lemma~\ref{g2ne}. Thus, by the assumption $G_2\vDash a\parallel b$, there must exist $e=\langle e_k\rangle_{k\in\mathbb Z}$ such that $e_a=e'_a=[0]_3$ and $e_{a+2}=e_b=e''_b=[2]_3$, which is a contradiction to Lemma~\ref{g2ne}.

\noindent$(\Leftarrow)$ Assume now that $f=\langle f_k\rangle_{k\in\mathbb Z}$ and  $g=\langle g_k\rangle_{k\in\mathbb Z}$ are two Nash equilibria of the game $G_2$. We need to show that there is an equilibrium $f=\langle f_k\rangle_{k\in\mathbb Z}$ of the same game $G_2$ such that $e_a=f_a$ and $e_b=eg_b$. Note that by Lemma~\ref{g2ne}, $f_k,g_k\in\mathbb Z_3$ for each $k\in\mathbb Z$ and
\begin{equation}\label{e'}
f_{k+2}-f_{k}\in\{[0]_3,[1]_3\},
\end{equation}
\begin{equation}\label{e''}
g_{k+2}-g_{k}\in\{[0]_3,[1]_3\},
\end{equation}
 for each $k\in\mathbb Z$.
We will consider two separate cases: $b=a+1$ and $b>a+2$.

\begin{figure}[ht]
\begin{center}
\vspace{3mm}
\scalebox{.7}{\includegraphics{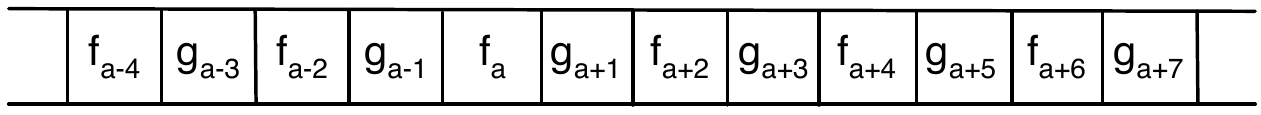}}
\vspace{0mm}
\caption{Nash equilibrium $e=\langle e_k\rangle_{k\in\mathbb Z}$.}\label{mainG2prop}
\vspace{0cm}
\end{center}
\vspace{0cm}
\end{figure}
\noindent{\em Case I}. Suppose that $b=a+1$. Consider (see Figure~\ref{mainG2prop}) strategy profile $e=\langle e_k\rangle_{k\in\mathbb Z}$ such that
$$
e_k =
\begin{cases}
f_k & \mbox{ if $k\equiv a \pmod 2$},\\
g_k & \mbox{ if $k\equiv a+1 \pmod 2$}.
\end{cases}
$$
Note that $e_{k+2}-e_{k}\in\{[0]_3,[1]_3\}$ for each $k\in\mathbb Z$ due statements (\ref{e'}) and (\ref{e''}). Thus, by Lemma~\ref{g2ne}, strategy profile $e$ is a Nash equilibrium of game $G_2$. Note that $e_a=f_a$ and $e_b=e_{a+1}=g_{a+1}=g_b$.

\noindent{\em Case II}. Assume now that $b>a+2$. Thus, $b-a>2$. Hence, by Lemma~\ref{little lemma}, there must exists $z_a,z_{a+1},\dots,z_b\in\mathbb Z_3$ such that $z_a=f_a$, $z_b=g_b$, and $z_{k+2}-z_k\in \{[0]_3,[1]_3\}$ for each $k$ such that $a\le k\le b-2$. Consider strategy profile $e=\langle e_k\rangle_{k\in\mathbb Z}$ such that
$$
e_k =
\begin{cases}
z_a & \mbox{ if $k< a$},\\
z_k & \mbox{ if $a\le k\le b$},\\
z_b & \mbox{ if $b<k$}.
\end{cases}
$$
By Lemma~\ref{g2ne}, strategy profile $e$ is a Nash equilibrium of the game $G_2$. Note that $e_a=z_a=f_a$ and $e_b=z_b=g_b$, by the choice of the sequence $z_a,z_{a+1},\dots,z_b$.
\end{proof}

\subsection{Game $G_n$: general case}\label{gn}

In this section we define game $G_n$ for $n\ge 3$. The set of strategies $S^n$ of the game $G_n$ is $(\mathbb Z_{n+1} \times \mathbb Z_{n+1})^{n-1}$. We visually represent elements of $S^n$ as $(n-1)\times 2$ matrices whose elements belong to $\mathbb Z_{n+1}$.

\begin{definition}\label{gn}
Pay-off function
\begin{equation}\label{u}
u\left(\left(
\begin{array}{ll}
x_{1,1} & x_{1,2}\\
x_{2,1} & x_{2,2} \\
\vdots & \vdots  \\
x_{n-1,1} & x_{n-1,2} 
\end{array}
\right),
\left(
\begin{array}{ll}
y_{1,1} & y_{1,2}  \\
y_{2,1} & y_{2,2} \\
\vdots & \vdots \\
y_{n-1,1} & y_{n-1,2} 
\end{array}
\right),
\left(
\begin{array}{ll}
z_{1,1} & z_{1,2} \\
z_{2,1} & z_{2,2} \\
\vdots & \vdots  \\
z_{n-1,1} & z_{n-1,2} 
\end{array}
\right)\right)
\end{equation}
 is equal to 1 if the following conditions are satisfied:
\begin{enumerate}
\item $y_{1,1}=[0]_{n+1}$,
\item $y_{k+1,2}+z_{k+1,1}-x_{k,2}-y_{k,1} \in \{[0]_{n+1},[1]_{n+1}\}$,
for every $1\le k<n-1$, 
\item $z_{1,2} - x_{n-1,2}-y_{n-1,1}\in\{[0]_{n+1},[1]_{n+1}\}$.
\end{enumerate}
if at least one of the above conditions is not satisfied, then pay-off function (\ref{u}) is equal to 0.
\end{definition}

\subsubsection{Perfect strategy profiles}

While describing properties of game $G_n$, it will be convenient to use terms ``perfect strategy profile" and ``semi-perfect strategy profile" at a particular player. We introduce the notion of a perfect profile in this section and the notion of a semi-perfect profile in the next section. 

\begin{definition}\label{perfect}
Strategy profile 
$$\scaleleftright[2ex]{\langle}{\left(
\begin{array}{ll}
x^i_{1,1} & x^i_{1,2} \\ \relax
x^i_{2,1} & x^i_{2,2} \\ \relax
\vdots & \vdots \\ \relax
x^i_{n-1,1} & x^i_{n-1,2} 
\end{array}
\right)}{\rangle}_{i\in \mathbb Z}$$
of the game $G_n$, where $n\ge 3$, is perfect at player $i$  if
\begin{enumerate}
\item $x^i_{1,1}=[0]_{n+1}$,
\item $x^{i-1}_{k,2}+x^i_{k,1} =x^i_{k+1,2}+x^{i+1}_{k+1,1}$ in $\mathbb Z_{n+1}$ for every $1\le k<n-1$,
\item  $x^{i-1}_{n-1,2}+x^i_{n-1,1}=x^{i+1}_{1,2}$ in $\mathbb Z_{n+1}$.
\end{enumerate}
\end{definition}

By a sum of two strategies in the game $G_n$, where $n\ge 3$, we mean element-wise sum of the two matrices. 
\begin{lemma}\label{perfect sum}
For any two strategy profiles $\langle s_i\rangle_{i\in\mathbb Z}$and $\langle s'_i\rangle_{i\in\mathbb Z}$, of the game $G_n$ perfect at a player $i\in\mathbb Z$, strategy profile $\langle s_i+s'_i\rangle_{i\in\mathbb Z}$ is also perfect at player $i$. \qed
\end{lemma}
\begin{proof}
See Definition~\ref{perfect}.
\end{proof}

\begin{lemma}\label{perfect exists}
For any 
$$
M =\left(
\begin{array}{ll}
[0] & x_{1,2} \\ \relax
x_{2,1} & x_{2,2} \\ \relax
x_{3,1} & x_{3,2} \\ \relax
\vdots & \vdots \\ \relax
x_{n-2,1} & x_{n-2,2} \\ \relax
x_{n-1,1} & x_{n-1,2} 
\end{array}
\right)
\in (\mathbb Z_{n+1} \times \mathbb Z_{n+1})^{n-1}
$$
and any $a,b\in \mathbb Z$, if  $0< |a-b| < n$, then there is a strategy profile $s=\langle s_i\rangle_{i\in\mathbb Z}$ of the game $G_n$ such that 
\begin{enumerate}
\item $s_a = M$,
\item $s_b =
\left(
\begin{array}{ll}
[0] & [0] \\ \relax
[0] & [0] \\ \relax
[0] & [0] \\ \relax
\vdots & \vdots \\ \relax
[0] & [0] \\ \relax
[0] & [0] 
\end{array}
\right),$
\item strategy profile $s$ is perfect at each player $i$ such that $a<i<b$.
\end{enumerate}
\end{lemma}
\begin{proof}
We assume that $a<b$. The case $b<a$ could be shown in a similar way. 

\noindent{\em Case $b=a+1$:} Consider strategy profile $s=\langle s_i\rangle_{i\in\mathbb Z}$ such that $s_a=M$ and all other strategy are equal to 
$$ \left(
\begin{array}{ll}
[0] & [0] \\ \relax
[0] & [0] \\ \relax
[0] & [0] \\ \relax
\vdots & \vdots \\ \relax
[0] & [0] \\ \relax
[0] & [0] 
\end{array}
\right).$$
The third condition of the lemma is satisfied vacuously.

\noindent{\em Case $a+1<b<a+n$:} Consider strategy profile $s=\langle s_i\rangle_{i\in\mathbb Z}$ such that strategies $s_a,s_{a+1},\dots,s_b$ are equal to 
$$\left(
\begin{array}{ll}
[0] & x_{1,2}  \\ \relax
x_{2,1} & x_{2,2}  \\ \relax
x_{3,1} & x_{3,2}  \\ \relax
x_{4,1} & x_{4,2}  \\ \relax
\vdots & \vdots \\ \relax
x_{k-1,1} & x_{k-1,2}  \\ \relax
x_{k,1} & x_{k,2}  \\ \relax
x_{k+1,1} & x_{k+1,2}  \\ \relax
x_{k+2,1} & x_{k+2,2}  \\ \relax
\vdots & \vdots \\ \relax
x_{n-3,1} & x_{n-3,2} \\ \relax
x_{n-2,1} & x_{n-2,2} \\ \relax
x_{n-1,1} & x_{n-1,2}
\end{array}
\right),
\left(
\begin{array}{ll}
[0] & [0]  \\ \relax
-x_{2,2} & x_{1,2}  \\ \relax
-x_{3,2} & [0]  \\ \relax
-x_{4,2} & [0]  \\ \relax
\vdots & \vdots \\ \relax
-x_{k-1,2} & [0]  \\ \relax
-x_{k,2} & [0]  \\ \relax
-x_{k+1,2} & [0]  \\ \relax
-x_{k+2,2} & [0]  \\ \relax
\vdots & \vdots \\ \relax
-x_{n-3,2} & [0] \\ \relax
-x_{n-2,2} & [0] \\ \relax
-x_{n-1,2} & [0] 
\end{array}
\right),
\left(
\begin{array}{ll}
[0] & [0]  \\ \relax
[0] & [0]  \\ \relax
[0] & x_{1,2}  \\ \relax
[0] & [0]  \\ \relax
\vdots & \vdots \\ \relax
[0] & [0]  \\ \relax
[0] & [0]  \\ \relax
[0] & [0]  \\ \relax
[0] & [0]  \\ \relax
\vdots & \vdots \\ \relax
[0] & [0]  \\ \relax
[0] & [0]  \\ \relax
[0] & [0] 
\end{array}
\right), 
$$
$$
\dots,
\left(
\begin{array}{ll}
[0] & [0]  \\ \relax
[0] & [0]  \\ \relax
[0] & [0]  \\ \relax
[0] & [0]  \\ \relax
\vdots & \vdots \\ \relax
[0] & x_{1,2}  \\ \relax
[0] & [0]  \\ \relax
[0] & [0]  \\ \relax
[0] & [0]  \\ \relax
\vdots & \vdots \\ \relax
[0] & [0]  \\ \relax
[0] & [0]  \\ \relax
[0] & [0] 
\end{array}
\right),
\left(
\begin{array}{ll}
[0] & [0]  \\ \relax
[0] & [0]  \\ \relax
[0] & [0]  \\ \relax
[0] & [0]  \\ \relax
\vdots & \vdots \\ \relax
[0] & [0]  \\ \relax
[0] & x_{1,2}  \\ \relax
[0] & [0]  \\ \relax
[0] & [0]  \\ \relax
\vdots & \vdots \\ \relax
[0] & [0]  \\ \relax
[0] & [0]  \\ \relax
[0] & [0] 
\end{array}
\right),
\left(
\begin{array}{ll}
[0] & [0]  \\ \relax
[0] & [0]  \\ \relax
[0] & [0]  \\ \relax
[0] & [0]  \\ \relax
\vdots & \vdots \\ \relax
[0] & [0]  \\ \relax
[0] & [0]  \\ \relax
[0] & [0]  \\ \relax
[0] & [0]  \\ \relax
\vdots & \vdots \\ \relax
[0] & [0]  \\ \relax
[0] & [0]  \\ \relax
[0] & [0] 
\end{array}
\right)
$$
respectively, where $k=b-a$. All other strategies are equal to
$$ \left(
\begin{array}{ll}
[0] & [0]  \\ \relax
[0] & [0]  \\ \relax
[0] & [0]  \\ \relax
[0] & [0]  \\ \relax
\vdots & \vdots \\ \relax
[0] & [0]  \\ \relax
[0] & [0]  \\ \relax
[0] & [0]  \\ \relax
[0] & [0]  \\ \relax
\vdots & \vdots \\ \relax
[0] & [0]  \\ \relax
[0] & [0]  \\ \relax
[0] & [0] 
\end{array}
\right).$$
By Definition~\ref{perfect}, this profile is perfect at each player $i$ such that $a<i<b$.
\end{proof}

\subsubsection{Semi-perfect strategy profiles}

\begin{definition}\label{semi-perfect}
Strategy profile
$$\scaleleftright[2ex]{\langle}{\left(
\begin{array}{ll}
x^i_{1,1} & x^i_{1,2} \\ \relax
x^i_{2,1} & x^i_{2,2} \\ \relax
\vdots & \vdots \\ \relax
x^i_{n-1,1} & x^i_{n-1,2} 
\end{array}
\right)}{\rangle}_{i\in \mathbb Z}$$
of the game $G_n$, where $n\ge 2$, is semi-perfect at player $i$ if
\begin{enumerate}
\item $x^i_{1,1}=[0]_{n+1}$,\label{semi-perfect-line-1}
\item  $x^i_{k+1,2}+x^{i+1}_{k+1,1}-x^{i-1}_{k,2}-x^i_{k,1}\in\{[0]_{n+1},[1]_{n+1}\}$, for every $1\le k<n-1$,\label{semi-perfect-line-2}
\item $x^{i+1}_{1,2} - x^{i-1}_{n-1,2}-x^i_{n-1,1}\in\{[0]_{n+1},[1]_{n+1}\}$.\label{semi-perfect-line-3}
\end{enumerate}
\end{definition}

\begin{lemma}\label{semi-perfect exists}
For any 
$$
A =\left(
\begin{array}{ll}
[0] & x_{1,2} \\ \relax
x_{2,1} & x_{2,2} \\ \relax
\vdots & \vdots \\ \relax
x_{n-1,1} & x_{n-1,2} 
\end{array}
\right)
\in (\mathbb Z_{n+1} \times \mathbb Z_{n+1})^{n-1},
$$
any 
$$
B =\left(
\begin{array}{ll}
[0] & y_{1,2} \\ \relax
y_{2,1} & y_{2,2} \\ \relax
\vdots & \vdots \\ \relax
y_{n-1,1} & y_{n-1,2} 
\end{array}
\right)
\in (\mathbb Z_{n+1} \times \mathbb Z_{n+1})^{n-1},
$$
and any $a,b\in \mathbb Z$, if  $|a-b|>n$, then there is a strategy profile $s=\langle s_i\rangle_{i\in\mathbb Z}$ of the game $G_n$ such that 
\begin{enumerate}
\item $s_a = A$,
\item $s_b = B$,
\item strategy profile $s$ is semi-perfect at each player $i$ such that $a<i<b$.
\end{enumerate}
\end{lemma}
\begin{proof}
We will assume that $a<b$. The other case is similar. Thus, $b-a>n$ due to the assumption $|a-b|>n$. 
Let $k$ be an integer such that $0 \le k< n-1$ and $k\equiv b - a\pmod{n-1}$. 
We first consider case when $k\neq 1$. 
Since $x_{1,2},y_{k,1}\in{\mathbb Z}_{n+1}$, by Lemma~\ref{little lemma}, there must exist a sequence  $z_1,z_2,\dots,z_{b-a}$ of equivalence classes in $\mathbb Z_{n+1}$ such that 
\begin{align}
& z_1=x_{1,2},\notag\\
& z_{b-a}=y_{k,1},\label{y rule}\\
& z_{i+1}-z_i\in\{[0]_{n+1},[1]_{n+1}\}  \hspace{1cm} (0\le i< b - a).\label{z rule}
\end{align} 
Consider now strategy profile $s=\langle s_i\rangle_{i\in\mathbb Z}$ such that strategies $s_a,s_{a+1},\dots,s_b$ are equal to 
$$\left(
\begin{array}{ll}
[0] & z_1  \\ \relax
x_{2,1} & x_{2,2}  \\ \relax
x_{3,1} & x_{3,2}  \\ \relax
x_{4,1} & x_{4,2}  \\ \relax
\vdots & \vdots \\ \relax
x_{k-1,1} & x_{k-1,2}  \\ \relax
x_{k,1} & x_{k,2}  \\ \relax
x_{k+1,1} & x_{k+1,2}  \\ \relax
\vdots & \vdots \\ \relax
x_{n-3,1} & x_{n-3,2} \\ \relax
x_{n-2,1} & x_{n-2,2} \\ \relax
x_{n-1,1} & x_{n-1,2}
\end{array}
\right),
\left(
\begin{array}{ll}
[0] & [0]  \\ \relax
-x_{2,2} & [0]  \\ \relax
-x_{3,2} & [0]  \\ \relax
-x_{4,2} & [0]  \\ \relax
\vdots & \vdots \\ \relax
-x_{k-1,2} & [0]  \\ \relax
-x_{k,2} & [0]  \\ \relax
-x_{k+1,2} & [0]  \\ \relax
\vdots & \vdots \\ \relax
-x_{n-3,2} & [0] \\ \relax
-x_{n-2,2} & [0] \\ \relax
-x_{n-1,2} & [0] 
\end{array}
\right),
\left(
\begin{array}{ll}
[0] & [0]  \\ \relax
z_2 & [0]  \\ \relax
[0] & [0]  \\ \relax
[0] & [0]  \\ \relax
\vdots & \vdots \\ \relax
[0] & [0]  \\ \relax
[0] & [0]  \\ \relax
[0] & [0]  \\ \relax
\vdots & \vdots \\ \relax
[0] & [0]  \\ \relax
[0] & [0]  \\ \relax
[0] & [0] 
\end{array}
\right),
\left(
\begin{array}{ll}
[0] & [0]  \\ \relax
[0] & [0]  \\ \relax
z_3 & [0]  \\ \relax
[0] & [0]  \\ \relax
\vdots & \vdots \\ \relax
[0] & [0]  \\ \relax
[0] & [0]  \\ \relax
[0] & [0]  \\ \relax
\vdots & \vdots \\ \relax
[0] & [0]  \\ \relax
[0] & [0]  \\ \relax
[0] & [0] 
\end{array}
\right)
,\dots,
$$
$$
\left(
\begin{array}{ll}
[0] & [0]  \\ \relax
[0] & [0]  \\ \relax
[0] & [0]  \\ \relax
[0] & [0]  \\ \relax
\vdots & \vdots \\ \relax
[0] & [0]  \\ \relax
[0] & [0]  \\ \relax
[0] & [0]  \\ \relax
\vdots & \vdots \\ \relax
[0] & [0]  \\ \relax
z_{n-2} & [0]  \\ \relax
[0] & [0] 
\end{array}
\right),
\left(
\begin{array}{ll}
[0] & [0]  \\ \relax
[0] & [0]  \\ \relax
[0] & [0]  \\ \relax
[0] & [0]  \\ \relax
\vdots & \vdots \\ \relax
[0] & [0]  \\ \relax
[0] & [0]  \\ \relax
[0] & [0]  \\ \relax
\vdots & \vdots \\ \relax
[0] & [0]  \\ \relax
[0] & [0]  \\ \relax
z_{n-1} & [0] 
\end{array}
\right),
\left(
\begin{array}{ll}
[0] & z_n  \\ \relax
[0] & [0]  \\ \relax
[0] & [0]  \\ \relax
[0] & [0]  \\ \relax
\vdots & \vdots \\ \relax
[0] & [0]  \\ \relax
[0] & [0]  \\ \relax
[0] & [0]  \\ \relax
\vdots & \vdots \\ \relax
[0] & [0]  \\ \relax
[0] & [0]  \\ \relax
[0] & [0] 
\end{array}
\right),
\left(
\begin{array}{ll}
[0] & [0]  \\ \relax
[0] & [0]  \\ \relax
[0] & [0]  \\ \relax
[0] & [0]  \\ \relax
\vdots & \vdots \\ \relax
[0] & [0]  \\ \relax
[0] & [0]  \\ \relax
[0] & [0]  \\ \relax
\vdots & \vdots \\ \relax
[0] & [0]  \\ \relax
[0] & [0]  \\ \relax
[0] & [0] 
\end{array}
\right),
\left(
\begin{array}{ll}
[0] & [0]  \\ \relax
z_{n+1} &  [0] \\ \relax
[0] & [0]  \\ \relax
[0] & [0]  \\ \relax
\vdots & \vdots \\ \relax
[0] & [0]  \\ \relax
[0] & [0]  \\ \relax
[0] & [0]  \\ \relax
\vdots & \vdots \\ \relax
[0] & [0]  \\ \relax
[0] & [0]  \\ \relax
[0] & [0] 
\end{array}
\right),
$$
$$
\left(
\begin{array}{ll}
[0] & [0]  \\ \relax
[0] &  [0] \\ \relax
z_{n+2} & [0]  \\ \relax
[0] & [0]  \\ \relax
\vdots & \vdots \\ \relax
[0] & [0]  \\ \relax
[0] & [0]  \\ \relax
[0] & [0]  \\ \relax
\vdots & \vdots \\ \relax
[0] & [0]  \\ \relax
[0] & [0]  \\ \relax
[0] & [0] 
\end{array}
\right),
\dots,
\left(
\begin{array}{ll}
[0] & [0]  \\ \relax
[0] &  -y_{2,1} \\ \relax
[0] & -y_{3,1}  \\ \relax
[0] & -y_{4,1}  \\ \relax
\vdots & \vdots \\ \relax
z_{b-a-1} & -y_{k-1,1}  \\ \relax
[0] & [0]  \\ \relax
[0] & -y_{k+1,1}  \\ \relax
\vdots & \vdots \\ \relax
[0] & -y_{n-3,1}  \\ \relax
[0] & -y_{n-2,1}  \\ \relax
[0] & -y_{n-1,1} 
\end{array}
\right),
\left(
\begin{array}{ll}
[0] & y_{1,2}  \\ \relax
y_{2,1} & y_{2,2}  \\ \relax
y_{3,1} & y_{3,2}  \\ \relax
y_{4,1} & y_{4,2}  \\ \relax
\vdots & \vdots \\ \relax
y_{k-1,1} & y_{k-1,2}  \\ \relax
z_{b-a} & y_{k,2}  \\ \relax
y_{k+1,1} & y_{k+1,2}  \\ \relax
\vdots & \vdots \\ \relax
y_{n-3,1} & y_{n-3,2} \\ \relax
y_{n-2,1} & y_{n-2,2} \\ \relax
y_{n-1,1} & y_{n-1,2}
\end{array}
\right)
$$
respectively. All other strategies are equal 
$$ \left(
\begin{array}{ll}
[0] & [0]  \\ \relax
[0] & [0]  \\ \relax
[0] & [0]  \\ \relax
[0] & [0]  \\ \relax
\vdots & \vdots \\ \relax
[0] & [0]  \\ \relax
[0] & [0]  \\ \relax
[0] & [0]  \\ \relax
\vdots & \vdots \\ \relax
[0] & [0]  \\ \relax
[0] & [0]  \\ \relax
[0] & [0] 
\end{array}
\right).$$
Due to condition~(\ref{z rule}), this profile is semi-perfect at each player $i$ such that $a<i<b$. Case $k=1$ is similar except that equation  (\ref{y rule}) should be replaced with $z_{b-a}=y_{k,2}$.
\end{proof}

\begin{lemma}[right expansion]\label{semi-perfect right expansion}
For each $a,b\in \mathbb Z$ such that $a<b$, if strategy profile $s=\langle s_i\rangle_{i\in\mathbb Z}$ is semi-perfect for each player $i$ such that $a<i<b$, then there is a strategy profile $s'=\langle s'_i\rangle_{i\in\mathbb Z}$ such that
\begin{enumerate}
\item $s'_i=s_i$ for each $i$ such that $a\le i\le b$,
\item $s'$ is semi-perfect for each player $i$ such that $a<i<b+1$.
\end{enumerate}
\end{lemma}
\begin{proof}
Let 
$$s_{b-1}=
\left(
\begin{array}{ll}
[0] & x_{1,2} \\ \relax
x_{2,1} & x_{2,2} \\ \relax
x_{3,1} & x_{3,2} \\ \relax
\vdots & \vdots \\ \relax
x_{n-2,1} & x_{n-2,2} \\ \relax
x_{n-1,1} & x_{n-1,2} 
\end{array}
\right),
\hspace{5mm}
s_{b}=
\left(
\begin{array}{ll}
[0] & y_{1,2} \\ \relax
y_{2,1} & y_{2,2} \\ \relax
y_{3,1} & y_{3,2} \\ \relax
\vdots & \vdots \\ \relax
y_{n-2,1} & y_{n-2,2} \\ \relax
y_{n-1,1} & y_{n-1,2} 
\end{array}
\right).
$$
Define $s'_i$ to be equal to $s_i$ for all $i\neq b+1$ and $s'_{b+1}$ to be 
$$
\left(
\begin{array}{ll}
[0] & x_{n-1,2}+y_{n-1,1} \\ \relax
x_{1,2}-y_{2,2} & [0] \\ \relax
x_{2,2}+y_{2,1}-y_{3,2} & [0] \\ \relax
\vdots & \vdots \\ \relax
x_{n-3,2}+y_{n-3,1}-y_{n-2,2} & [0] \\ \relax
x_{n-2,2}+y_{n-2,1}-y_{n-1,2} & [0] 
\end{array}
\right).
$$
By Definition~\ref{semi-perfect}, strategy profile $s'$ is perfect at player $b$.
\end{proof}

\begin{lemma}[left expansion]\label{semi-perfect left expansion}
For each $a,b\in \mathbb Z$ such that $a<b$, if strategy profile $s=\langle s_i\rangle_{i\in\mathbb Z}$ is semi-perfect for each player $i$ such that $a<i<b$, then there is a strategy profile $s'=\langle s'_i\rangle_{i\in\mathbb Z}$ such that
\begin{enumerate}
\item $s'_i=s_i$ for each $i$ such that $a\le i\le b$,
\item $s'$ is semi-perfect for each player $i$ such that $a-1<i<b$.
\end{enumerate}
\end{lemma}
\begin{proof}
Similar to the proof of Lemma~\ref{semi-perfect right expansion}.
\end{proof}

\begin{lemma}[infinite expansion]\label{semi-perfect expansion}
For each $a,b\in \mathbb Z$ such that $a<b$, if strategy profile $s=\langle s_i\rangle_{i\in\mathbb Z}$ is semi-perfect for each for player $i$ such that $a<i<b$, then there is a strategy profile $s'=\langle s'_i\rangle_{i\in\mathbb Z}$ such that
\begin{enumerate}
\item $s'_i=s_i$ for each $i$ such that $a\le i\le b$,
\item $s'$ is semi-perfect for each player $i\in\mathbb Z$.
\end{enumerate}
\end{lemma}
\begin{proof}
Follows from Lemma~\ref{semi-perfect right expansion} and Lemma~\ref{semi-perfect left expansion}.
\end{proof}

\subsubsection{Properties of Nash equilibria of game $G_n$}

\begin{lemma}\label{gn ne}
For any $n\ge 3$, strategy profile $e$ is a Nash equilibrium of the game $G_n$ if and only if profile $e$ is semi-perfect at each player $i\in \mathbb Z$.
\end{lemma}
\begin{proof}
See Definition~\ref{gn} and Definition~\ref{semi-perfect}.
\end{proof}

\begin{lemma}\label{diagonal lemma}
For any $a\in \mathbb Z$ and any $n\ge 3$, if 
$$\scaleleftright[2ex]{\langle}{\left(
\begin{array}{ll}
x^i_{1,1} & x^i_{1,2}  \\ \relax
x^i_{2,1} & x^i_{2,2}  \\ \relax
\vdots & \vdots \\ \relax
x^i_{n-1,1} & x^i_{n-1,2} 
\end{array}
\right)}{\rangle}_{i\in \mathbb Z}$$
is a Nash equilibrium of the game $G_n$ and $x^a_{1,2}=[0]_{n+1}$, then
$$
x^{a+k}_{k+1,2}+x^{a+k+1}_{k+1,1}\in\{[0],[1],[2],\dots,[k]\},
$$
for each $0\le k\le n-2$.
\end{lemma}
\begin{proof}
Induction on $k$. If $k=0$, then $x^{a+k}_{k+1,2}=x^a_{1,2}=[0]$ due to the assumption of the lemma. At the same time, $x^{a+k+1}_{k+1,1}=x^{a+1}_{1,1}=[0]$ by Lemma~\ref{gn ne} and item \ref{semi-perfect-line-1} of Definition~\ref{semi-perfect}. Thus,
$$x^{a+k}_{k+1,2}+x^{a+k+1}_{k+1,1}=[0]+[0]=[0]\in \{[0]\}.$$

For the induction step, assume that
$$
x^{a+k}_{k+1,2}+x^{a+k+1}_{k+1,1}\in\{[0],[1],[2],\dots,[k]\}.
$$
By Lemma~\ref{gn ne} and item \ref{semi-perfect-line-2} of Definition~\ref{semi-perfect}, there is $\varepsilon\in\{[0]_{n+1},[1]_{n+1}\}$ such that
$$
x^{a+k+1}_{k+3,2}+x^{a+k+2}_{k+2,1} = x^{a+k}_{k+1,2}+x^{a+k+1}_{k+1,1} + \varepsilon.
$$
Therefore,
$$
x^{a+k+1}_{k+3,2}+x^{a+k+2}_{k+2,1}\in \{[0],[1],[2],\dots,[k],[k+1]\}.
$$
\end{proof}

\begin{lemma}\label{post diagonal lemma}
For any $a\in \mathbb Z$ and any $n\ge 3$, if 
$$\scaleleftright[2ex]{\langle}{\left(
\begin{array}{ll}
x^i_{1,1} & x^i_{1,2}  \\ \relax
x^i_{2,1} & x^i_{2,2}  \\ \relax
\vdots & \vdots \\ \relax
x^i_{n-1,1} & x^i_{n-1,2} 
\end{array}
\right)}{\rangle}_{i\in \mathbb Z}$$
is a Nash equilibrium of the game $G_n$ and $x^a_{1,2}=[0]_{n+1}$, then
$$
x^{a+n}_{1,2}\in\{[0],[1],[2],\dots,[n-1]\}.
$$
\end{lemma}
\begin{proof}
By Lemma~\ref{diagonal lemma},
$$
x^{a+n-2}_{n-1,2}+x^{a+n-1}_{n-1,1}\in\{[0],[1],[2],\dots,[n-2]\}.
$$
By Lemma~\ref{gn ne} and item \ref{semi-perfect-line-3} of Definition~\ref{semi-perfect}, for $i=a+n-1$, there is $\varepsilon\in\{[0]_{n+1},[1]_{n+1}\}$ such that
$$x^{a+n}_{1,2}= x^{a+n-2}_{n-1,2}+x^{a+n-1}_{n-1,1}+\varepsilon.$$  
Therefore,
$x^{a+n}_{1,2}\in \{[0],[1],[2],\dots,[n-2],[n-1]\}$.
\end{proof}

\begin{definition}\label{fn}
For any $n\ge 3$, let $f^n=\langle f^n_i\rangle_{i\in\mathbb Z}$ be the strategy profile of the game $G_n$ such that
$$
f^n_i=
\left(
\begin{array}{ll}
[0]_{n+1} & [0]_{n+1} \\ \relax
[0]_{n+1} & [0]_{n+1} \\ \relax
\vdots & \vdots \\ \relax
[0]_{n+1} & [0]_{n+1}
\end{array}
\right),
$$
for each $i\in \mathbb Z$.
\end{definition}
\begin{lemma}\label{fn in ne}
$f^n\in NE(G_n)$ for each $n\ge 3$.
\end{lemma}
\begin{proof}
By Lemma~\ref{gn ne} and Definition~\ref{semi-perfect}.
\end{proof}

\begin{definition}\label{}
For any $n\ge 3$, let $g^n=\langle g^n_i\rangle_{i\in\mathbb Z}$ be the strategy profile of the game $G_n$ such that
$$
g^n_i=
\left(
\begin{array}{ll}
[0]_{n+1} & [n]_{n+1} \\ \relax
[0]_{n+1} & [n]_{n+1} \\ \relax
\vdots & \vdots \\ \relax
[0]_{n+1} & [n]_{n+1}
\end{array}
\right),
$$
for each $i\in\mathbb Z$.
\end{definition}
\begin{lemma}
$g^n\in NE(G_n)$ for each $n\ge 3$.
\end{lemma}
\begin{proof}
By Lemma~\ref{gn ne} and Definition~\ref{semi-perfect}.
\end{proof}

\subsubsection{Main property of game $G_n$}

\begin{theorem}\label{gn theorem}
$G_n\vDash a\parallel b$ if and only if $|a-b|\neq n$ and $a\neq b$, where $n\ge 3$.
\end{theorem}
\begin{proof}
\noindent$(\Rightarrow):$  Let $G_n\vDash a\parallel b$. First, suppose that $|a-b|=n$. Without loss of generality, assume that $b=a+n$. Due to assumption $G_n\vDash a\parallel b$, there must exist an equilibrium 
$$\langle e_i\rangle_{i\in \mathbb Z}
=\scaleleftright[2ex]{\langle}{\left(
\begin{array}{ll}
x^i_{1,1} & x^i_{1,2}  \\ \relax
x^i_{2,1} & x^i_{2,2}  \\ \relax
\vdots & \vdots \\ \relax
x^i_{n-1,1} & x^i_{n-1,2} 
\end{array}
\right)}{\rangle}_{i\in \mathbb Z}$$
of the game $G_n$ such that $e_a=f^n_a$ and $e_b=g^n_b$, where $f^n=\langle f^n_i\rangle_{i\in\mathbb Z}$  and $g^n=\langle g^n_i\rangle_{i\in\mathbb Z}$ are Nash equilibria of the game $G_n$ defined in the previous section. Thus, $x^a_{1,2}=[0]$ and $x^b_{1,2}=[n]$, which is a contradiction to Lemma~\ref{post diagonal lemma}.

Assume now that $a=b$. Due to assumption $G_n\vDash a\parallel b$, there must exist an equilibrium $e=\langle e_i\rangle_{i\in \mathbb Z}$ of the game $G_n$ such that $f^n_a=e_a=e_b=g^n_b$. Thus, $[0]_{n+1}=[n]_{n+1}$, which is a contradiction.

\noindent$(\Leftarrow):$ Without loss of generality, assume that $b>a$. 
To prove that $G_n\vDash a\parallel b$, consider any two Nash equilibria 
$$v=\langle v_i\rangle_{i\in \mathbb Z}=\scaleleftright[2ex]{\langle}{\left(
\begin{array}{ll}
[0] & x^i_{1,2}  \\ \relax
x^i_{2,1} & x^i_{2,2}  \\ \relax
\vdots & \vdots \\ \relax
x^i_{n-1,1} & x^i_{n-1,2} 
\end{array}
\right)}{\rangle}_{i\in \mathbb Z}$$
and
$$w=\langle w_i\rangle_{i\in \mathbb Z}=\scaleleftright[2ex]{\langle}{\left(
\begin{array}{ll}
[0] & y^i_{1,2}\\ \relax
y^i_{2,1} & y^i_{2,2} \\ \relax
\vdots & \vdots \\ \relax
y^i_{n-1,1} & y^i_{n-1,2}
\end{array}
\right)}{\rangle}_{i\in \mathbb Z}$$
of the game $G_n$. Note that the upper left element in each of the above matrices is $[0]$ due to Lemma~\ref{gn ne}. We need to show that there is an equilibrium $e\in NE(G_n)$ such that $e$ agrees with equilibrium $v$ on the strategy of player $a$ and with equilibrium $w$ on the strategy of player $b$. 

\noindent{\em Case I:} $b-a<n$. By Lemma~\ref{perfect exists}, there are strategy profiles $\langle s_i\rangle_{i\in \mathbb Z}$ and $\langle t_i\rangle_{i\in \mathbb Z}$ both perfect at each player $i$ such that $a<i<b$ and
$$s_a=
\left(
\begin{array}{ll}
[0] & x^a_{1,2}  \\ \relax
x^a_{2,1} & x^a_{2,2}  \\ \relax
\vdots & \vdots \\ \relax
x^a_{n-1,1} & x^a_{n-1,2} 
\end{array}
\right), 
\;\;\;\;\;
s_b=
\left(
\begin{array}{ll}
[0]_{n+1} & [0]_{n+1}  \\ \relax
[0]_{n+1} & [0]_{n+1}  \\ \relax
\vdots & \vdots \\ \relax
[0]_{n+1} & [0]_{n+1} 
\end{array}
\right),
$$
$$
t_a=
\left(
\begin{array}{ll}
[0]_{n+1} & [0]_{n+1}  \\ \relax
[0]_{n+1} & [0]_{n+1}  \\ \relax
\vdots & \vdots \\ \relax
[0]_{n+1} & [0]_{n+1} 
\end{array}
\right),\;\;\;\;\;
t_b=
\left(
\begin{array}{ll}
[0] & y^b_{1,2}  \\ \relax
y^b_{2,1} & y^b_{2,2}  \\ \relax
\vdots & \vdots \\ \relax
y^b_{n-1,1} & y^b_{n-1,2} 
\end{array}
\right).
$$
By Lemma~\ref{perfect sum}, strategy profile $\langle s_i+t_i\rangle_{i\in\mathbb Z}$ is perfect at each player $i$ such that $a<i<b$. Thus, by Lemma~\ref{semi-perfect expansion}, there is a strategy profile $\langle e_i\rangle_{i\in\mathbb Z}$ of the game $G_n$ where $e_i=s_i+t_i$ for each $i$ such that $a\le i\le b$ and $e_i$ is semi-perfect for each player $i\in\mathbb Z$. Hence, by Lemma~\ref{gn ne},  strategy profile $\langle e_i\rangle_{i\in\mathbb Z}$ is a Nash equilibrium of the game $G_n$. We are only left to notice that
\begin{align*}
&e_a=s_a+t_a=
\left(
\begin{array}{ll}
[0] & x^a_{1,2}  \\ \relax
x^a_{2,1} & x^a_{2,2}  \\ \relax
\vdots & \vdots \\ \relax
x^a_{n-1,1} & x^a_{n-1,2} 
\end{array}
\right)
+
\left(
\begin{array}{ll}
[0]_{n+1} & [0]_{n+1}  \\ \relax
[0]_{n+1} & [0]_{n+1}  \\ \relax
\vdots & \vdots \\ \relax
[0]_{n+1} & [0]_{n+1} 
\end{array}
\right)
\\
&=
\left(
\begin{array}{ll}
[0] & x^a_{1,2}  \\ \relax
x^a_{2,1} & x^a_{2,2}  \\ \relax
\vdots & \vdots \\ \relax
x^a_{n-1,1} & x^a_{n-1,2} 
\end{array}
\right)
=v_a
\end{align*}
and
\begin{align*}
&e_b=s_b+t_b=
\left(
\begin{array}{ll}
[0]_{n+1} & [0]_{n+1}  \\ \relax
[0]_{n+1} & [0]_{n+1}  \\ \relax
\vdots & \vdots \\ \relax
[0]_{n+1} & [0]_{n+1} 
\end{array}
\right)
+
\left(
\begin{array}{ll}
[0] & y^b_{1,2}  \\ \relax
y^b_{2,1} & y^b_{2,2}  \\ \relax
\vdots & \vdots \\ \relax
y^b_{n-1,1} & y^b_{n-1,2} 
\end{array}
\right)
\\
&=
\left(
\begin{array}{ll}
[0] & y^b_{1,2}  \\ \relax
y^b_{2,1} & y^b_{2,2}  \\ \relax
\vdots & \vdots \\ \relax
y^b_{n-1,1} & y^b_{n-1,2} 
\end{array}
\right)
=w_b.
\end{align*}

\noindent{\em Case II:} $b-a>n$. By Lemma~\ref{semi-perfect exists}, there exists a strategy profile $\langle s_i\rangle_{i\in \mathbb Z}$ such that $s_a=v_a$, $s_b=w_b$, and strategy profile $\langle s_i\rangle_{i\in \mathbb Z}$ is semi-perfect at each player $i$ such that $a<i<b$. By Lemma~\ref{semi-perfect expansion}, there is a strategy profile $\langle e_i\rangle_{i\in \mathbb Z}$ semi-perfect at each player $i\in \mathbb Z$ such that $e_i=s_i$ for each $a\le i\le b$. By Lemma~\ref{gn ne}, strategy profile $\langle e_i\rangle_{i\in \mathbb Z}$ is a Nash equilibrium of the game $G_n$. We are only left to notice that
$e_a=s_a=v_a$ and $e_b=s_b=w_b$.
\end{proof}

\subsection{Combining games together}\label{combine}

\begin{theorem}\label{Gall}
For any $n\ge 0$, and any $a,b\in \mathbb Z$, 
$$
G_n \Vdash a\parallel b \mbox {\hspace{1cm} if and only if \hspace{1cm}} |a-b| \neq n \mbox{ and } a\neq b.
$$
\end{theorem}
\begin{proof}
See Theorem~\ref{g0 theorem}, Theorem~\ref{g1 theorem}, Theorem~\ref{g2 theorem}, and Theorem~\ref{gn theorem}.
\end{proof}

\begin{theorem}\label{ne exists}
For any $n\ge 0$, game $G_n$ has at least one Nash equilibrium.
\end{theorem}
\begin{proof}
{\em Case I}: $n=0$. Consider strategy profile $e=\langle e_k\rangle_{k\in\mathbb Z}$ such that $e_k=[0]_2$ for each $k\in\mathbb Z$. By Lemma~\ref{g0ne}, strategy profile $e$ is a Nash equilibrium of the games $G_0$.

\noindent {\em Case II}: $n=1,2$. Consider strategy profile $e=\langle e_k\rangle_{k\in\mathbb Z}$ such that $e_k=[0]_3$ for each $k\in\mathbb Z$. By Lemma~\ref{g1ne} and Lemma~\ref{g2ne}, strategy profile $e$ is a Nash equilibrium of the games $G_1$ and $G_2$.

\noindent {\em Case III}: $n>2$. By Lemma~\ref{fn in ne},  strategy profile $f^n$ is a Nash equilibrium of the game $G_n$.
\end{proof}

\subsection{Product of games}\label{product section}

In this section we define product operation on cellular games. Informally, product is a composition of several cellular games played concurrently and independently. The pay-off of a player in the product of the games is the sum of the pay-offs in the individual games. 

\begin{definition}\label{product definition}
Product $\prod_{i=1}^n G^i$ of  any finite family of cellular games $\{G^i\}_{i=1}^n=\{(S^i,u^i)\}_{i=1}^n$ is the cellular game $G=(S,u)$, where
\begin{enumerate}
\item $S$ is Cartesian product $\prod_{i=1}^nS^i$,
\item $u(\langle x^i\rangle_{i\le n},\langle y^i\rangle_{i\le n},\langle z^i\rangle_{i\le n})=\sum_{i=1}^n u^i(x^i,y^i,z^i)$.
\end{enumerate}
\end{definition}

\begin{lemma}\label{product ne lemma}
If $\{G^i\}_{i=1}^n=\{(S^i,u^i)\}_{i=1}^n$ is a finite family of cellular games, then
$\langle\langle e^i_k\rangle_{i\le n}\rangle_{k\in \mathbb Z}\in NE(\prod_{i=1}^n G^i)$ if and only if
$\langle e^i_k\rangle_{k\in\mathbb Z}\in NE(G^i)$ for each $i\le n$.
\end{lemma}
\begin{proof}
$(\Rightarrow)$ Suppose that $e^{i_0}=\langle e^{i_0}_k\rangle_{k\in\mathbb Z}\notin NE(G^i)$ for some $i_0\le n$. Thus, there is $k_0\in\mathbb Z$ and $s\in S^{i_0}$ such that
\begin{equation}\label{eq alpha}
u^{i_0}(e^{i_0}_{k_0-1},s,e^{i_0}_{k_0+1}) > u^{i_0}(e^{i_0}_{k_0-1},e^{i_0}_{k_0},e^{i_0}_{k_0+1}).
\end{equation} 
Let tuple $\mathbf s$ be $\langle e^1_{k_0}, e^2_{k_0},\dots,  e^{i_0-1}_{k_0},s,e^{i_0+1}_{k_0},\dots, e^n_{k_0}\rangle$. Then, by Definition~\ref{product definition} and due to inequality~(\ref{eq alpha}),
\begin{eqnarray*}
u(\langle e^{i}_{k_0-1}\rangle_{i\le n},\mathbf s, \langle e^{i}_{k_0+1}\rangle_{i\le n})=
\sum_{i\neq i_0} u^i(e^{i}_{k_0-1}, e^i_{k_0}, e^{i}_{k_0+1}) +
 u^{i_0}(e^{i_0}_{k_0-1}, s, e^{i_0}_{k_0+1}) >\\
 \sum_{i\neq i_0} u^i(e^{i}_{k_0-1}, e^i_{k_0}, e^{i}_{k_0+1})  +
 u^{i_0}(e^{i_0}_{k_0-1},e^{i_0}_{k_0},e^{i_0}_{k_0+1}) = 
  \sum_{i} u^i(e^{i}_{k_0-1}, e^i_{k_0}, e^{i}_{k_0+1}) =\\
 = u(\langle e^{i}_{k_0-1}\rangle_{i\le n},\langle e^{i}_{k_0}\rangle_{i\le n}, \langle e^{i}_{k_0+1}\rangle_{i\le n}),
\end{eqnarray*}
which is a contradiction with the assumption $\langle\langle e^i_k\rangle_{i\le n}\rangle_{k\in \mathbb Z}\in NE(\prod_{i=1}^n G^i)$.

\noindent $(\Leftarrow)$ Assume now that 
$\langle\langle e^i_k\rangle_{i\le n}\rangle_{k\in \mathbb Z}\notin NE(\prod_{i=1}^n G^i)$. Thus, there must exist $k\in \mathbb Z$ and  ${\mathbf s} = \langle s^i\rangle_{i\le n}$ such that
$$
u(\langle e^{i}_{k_0-1}\rangle_{i\le n},\mathbf s, \langle e^{i}_{k_0+1}\rangle_{i\le n}) >
u(\langle e^{i}_{k_0-1}\rangle_{i\le n},\langle e^{i}_{k_0}\rangle_{i\le n}, \langle e^{i}_{k_0+1}\rangle_{i\le n}).
$$
Hence, by Definition~\ref{product definition},
$$
\sum_{i}u^i(e^{i}_{k_0-1},s^i, e^{i}_{k_0+1}) >
\sum_{i}u^i(e^{i}_{k_0-1},e^{i}_{k_0}, e^{i}_{k_0+1}).
$$
Thus, there must exist at least one $i_0\le n$ such that
$$
u^{i_0}(e^{i_0}_{k_0-1},s^{i_0}, e^{i_0}_{k_0+1}) >
u^{i_0}(e^{i_0}_{k_0-1},e^{i_0}_{k_0}, e^{i_0}_{k_0+1}),
$$
which is a contradiction with the assumption 
$\langle e^{i_0}_k\rangle_{k\in\mathbb Z}\in NE(G^{i_0})$.
\end{proof}

\begin{theorem}\label{product theorem}
If $a,b\in\mathbb Z$ and each of the cellular games in family $\{G^i\}_{i=1}^n=\{(S^i,u^i)\}_{i=1}^n$ has at least one Nash equilibrium, then $\prod_{i=1}^nG^i\vDash a\parallel b$ if and only if $G^i\vDash a\parallel b$ for each $i\le n$.
\end{theorem}
\begin{proof}
$(\Rightarrow)$ Let 
$\langle e^1_k\rangle_{k\in \mathbb Z}$,
$\langle e^2_k\rangle_{k\in \mathbb Z}$,
\dots,
$\langle e^n_k\rangle_{k\in \mathbb Z}$ be Nash equilibria of games $G^1,G^2,\dots,G^n$ that exist by the assumption of the theorem. Consider any $i_0\le n$ and any two equilibria  $\langle f_k\rangle_{k\in \mathbb Z}$ and $\langle g_k\rangle_{k\in \mathbb Z}$ of the game $G^{i_0}$. We need to show that there exists Nash equilibrium $h=\langle h_k\rangle_{k\in \mathbb Z}$ of the game $G^{i_0}$ such that $h_a=f_a$ and $h_b=g_b$. Indeed, by Theorem~\ref{product ne lemma}, 
$$
\langle\langle e^1_k,e^2_k,\dots,e^{i_0-1}_k,f_k,e^{i_0+1}_k,\dots,e^{n}_k\rangle\rangle_{k\in\mathbb Z}
$$
and
$$
\langle\langle e^1_k,e^2_k,\dots,e^{i_0-1}_k,g_k,e^{i_0+1}_k,\dots,e^{n}_k\rangle\rangle_{k\in\mathbb Z}
$$
are Nash equilibria of the game $\prod_{i=1}^nG^i$. Hence, by the assumption $\prod_{i=1}^nG^i\vDash a\parallel b$, there exists a Nash equilibrium 
$$
\langle\langle w^1_k,w^2_k,\dots,w^{i_0-1}_k,w^{i_0}_k,w^{i_0+1}_k,\dots,w^{n}_k\rangle\rangle_{k\in\mathbb Z}
$$
 of the game $\prod_{i=1}^nG^i$ such that
\begin{equation*}
\langle w^1_a,\dots,w^{i_0-1}_a,w^{i_0}_a,w^{i_0+1}_a,\dots,w^{n}_a\rangle =
\langle e^1_a,\dots,e^{i_0-1}_a,f_a,e^{i_0+1}_a,\dots,e^{n}_a\rangle
\end{equation*}
and
\begin{equation*}
\langle w^1_b,\dots,w^{i_0-1}_b,w^{i_0}_b,w^{i_0+1}_b,\dots,w^{n}_b\rangle =
\langle e^1_b,\dots,e^{i_0-1}_b,f_b,e^{i_0+1}_b,\dots,e^{n}_b\rangle.
\end{equation*}
In particular, $w^{i_0}_a=f_a$ and $w^{i_0}_b=g_b$. At the same time, by Theorem~\ref{product ne lemma}, $\langle w^{i_0}_k\rangle_{k\in\mathbb Z}$ is a Nash equilibrium of game $G^{i_0}$. Let $h$ be $\langle w^{i_0}_k\rangle_{k\in\mathbb Z}$.

\noindent$(\Leftarrow)$ Let 
$\langle\langle f^1_k,f^2_k,\dots,f^n_k\rangle\rangle_{k\in\mathbb Z}$
and
$\langle\langle g^1_k,g^2_k,\dots,g^n_k\rangle\rangle_{k\in\mathbb Z}$
be two Nash equilibria of the game $\prod_{i=1}^nG^i$. We will prove that there is a Nash equilibrium 
$\langle\langle h^1_k,h^2_k,\dots,h^n_k\rangle\rangle_{k\in\mathbb Z}$
of game $\prod_{i=1}^nG^i$ such that
$$
\langle h^1_a,h^2_a,\dots,h^n_a\rangle=\langle f^1_a,f^2_a,\dots,f^n_a\rangle 
$$
and
$$
\langle h^1_b,h^2_b,\dots,h^n_b\rangle=\langle g^1_b,g^2_b,\dots,g^n_b\rangle. 
$$
Indeed, by Theorem~\ref{product ne lemma}, strategy profiles $\langle f^i_k\rangle_{k\in\mathbb Z}$ and $\langle g^i_k\rangle_{k\in\mathbb Z}$ are Nash equilibria of game $G^i$ for each $i\le n$. Thus, due to the assumption of the theorem, for each $i\le n$ there exists Nash equilibrium  $\langle h^i_k\rangle_{k\in\mathbb Z}$ of the game $G^i$ such that $h^i_a=f^i_a$ and $h^i_b=g^i_b$ for each $i\le n$. By Theorem~\ref{product ne lemma}, 
$
\langle\langle h^1_k,h^2_k,\dots,h^n_k\rangle\rangle_{k\in\mathbb Z}
$
is a Nash equilibrium of the game $\prod_{i=1}^nG^i$.
\end{proof}

\subsection{Game $G_\infty$}\label{ginf}

\begin{definition}\label{ginf def}
Game $G_\infty$ is pair $(\{0\},0)$, whose first element is the single-element set containing number $0$ and whose second component is the constant function equal to $0$.
\end{definition}
\begin{lemma}\label{ginf lemma}
$G_\infty\vDash a\parallel b$ for each $a,b\in\mathbb Z$.
\end{lemma}
\begin{proof}
Game $G_\infty$ has a unique strategy profile, which is also the unique Nash equilibrium of this game.
\end{proof}

\subsection{Completeness: final steps}\label{final steps}

We are now ready to state and prove the completeness theorem for our logical system.

\begin{theorem}[completeness]\label{completeness}
For each formula $\phi\in\Phi$, if $G\vDash\phi$ for each cellular game $G$, then $\vdash\phi$.
\end{theorem}
\begin{proof}
Suppose that $\nvdash\phi$. Let $M$ be any maximal consistent subset of $\Phi$ such that $\neg\phi\in M$. Thus, $\phi\notin M$ due to assumption of the consistency of $M$. There are two cases that we consider separately.

\noindent{\em Case I}: $0\parallel 0\in M$. 
\begin{lemma}\label{}
$\psi\in M$ if and only if $G_\infty\vDash\psi$ for each $\psi\in\Phi$.
\end{lemma}
\begin{proof}
Induction on structural complexity of formula $\psi$. If $\phi$ is $\bot$, then $\psi\notin M$ due to consistency of set $M$. At the same time, $G_\infty\nvDash\psi$ by Definition~\ref{true}.

Suppose now that $\psi$ is formula $a\parallel b$. By Reflexivity axiom, $0\parallel 0 \to a\parallel b$.  Thus, $a\parallel b\in M$, by the assumption $0\parallel 0\in M$ and due to maximality of set $M$. At the same time, $G_\infty\vDash a\parallel b$ by Lemma~\ref{ginf lemma}.

Finally, case when $\psi$ is an implication $\sigma\to\tau$ follows in the standard way from maximality and consistency of set $M$ and Definition~\ref{true}.
\end{proof}
To finish the proof of the theorem, note that $\phi\notin M$ implies, by the above lemma, that $G_\infty\nvDash\phi$.

\noindent{\em Case II}: $0\parallel 0\notin M$. 
Let $Sub(\phi)$ be the finite set of all $(p,q)\in {\mathbb Z}^2$ such that $p\parallel q$ is a subformula of $\phi$. Define game $G$ to be
$$
\prod \{G_{|p-q|} \;|\; \mbox{$(p, q)\in Sub(\phi)$ and $p\parallel q \notin M$}\}.
$$
\begin{lemma}\label{sub lemma}
For each subformula $\psi$ of formula $\phi$,
$$\psi\in M \mbox{\;\;\;\;\;if and only if\;\;\;\;\;} G\vDash\psi.$$ 
\end{lemma}
\begin{proof}
Induction on structural complexity of formula $\psi$. 
If $\phi$ is $\bot$, then $\psi\notin M$ due to consistency of set $M$. At the same time, $G\nvDash\psi$ by Definition~\ref{true}.

Next assume that $\psi$ is an atomic formula $a\parallel b$. 

\noindent$(\Rightarrow):$ Suppose first that $a\parallel b\in M$ and $G\nvDash a\parallel b$. Thus, by Theorem~\ref{product theorem} and Theorem~\ref{ne exists}, there must exist $(p,q)\in Sub(\phi)$ such that $p\parallel q\notin M$ and $G_{|p-q|}\nvDash a\parallel b$. Hence, by Theorem~\ref{Gall}, either $a=b$ or $|p-q|=|a-b|$. If $a=b$, then assumption $a\parallel b\in M$ implies that $0\parallel 0\in M$ due to Homogeneity axiom and maximality of set $M$. The latter, however is contradiction with the assumption of the case that we consider. If $|p-q|=|a-b|$, then $a\parallel b\in M$ implies $p\parallel q\in M$ by Lemma~\ref{abs value lemma}. The latter is a contradiction with the choice of pair $(p,q)$.

\noindent$(\Leftarrow):$ Suppose that $a\parallel b\notin M$. Thus, $G_{|a-b|}\nvDash a\parallel b$ by Theorem~\ref{Gall}.  Hence, $G\nvDash a\parallel b$ by Theorem~\ref{product theorem} and Theorem~\ref{ne exists}. 

Case when $\psi$ is an implication $\sigma\to\tau$ follows in the standard way from maximality and consistency of set $M$ and Definition~\ref{true}.
\end{proof}
To finish the proof of the theorem, note that $\phi\notin M$ implies, by Lemma~\ref{sub lemma}, that $G\nvDash\phi$.
\end{proof}

\section{Conclusion}\label{conclusion}

In this paper we have shown existence of cellular games in which Nash equilibria are interchangeable for near-by players, but not interchangeable for far-away players. We also gave complete axiomatization of all propositional properties of interchangeability of cellular games. Possible next step in this work could be axiomatization of properties of Nash equilibria for two-dimensional or even circular cellular games. Circular economies has been studies in the economics literature before~\cite{k74jet}.

\bibliography{../sp}

\begin{thebibliography}{10}

\bibitem{gpp91}
Dan Geiger, Azaria Paz, and Judea Pearl.
\newblock Axioms and algorithms for inferences involving probabilistic
  independence.
\newblock {\em Inform. and Comput.}, 91(1):128--141, 1991.

\bibitem{hn13lori}
Kristine Harjes and Pavel Naumov.
\newblock Cellular games, nash equilibria, and {F}ibonacci numbers.
\newblock In {\em The Fourth International Workshop on Logic, Rationality and
  Interaction (LORI-IV)}, pages 149--161. Springer, 2013.

\bibitem{k74jet}
Marvin Kraus.
\newblock Land use in a circular city.
\newblock {\em Journal of Economic Theory}, 8(4):440 -- 457, 1974.

\bibitem{mn10}
Sara Miner~More and Pavel Naumov.
\newblock An independence relation for sets of secrets.
\newblock {\em Studia Logica}, 94(1):73--85, 2010.

\bibitem{mns11csl}
Sara~Miner More, Pavel Naumov, and Benjamin Sapp.
\newblock Concurrency semantics for the {G}eiger-{P}az-{P}earl axioms of
  independence.
\newblock In Marc Bezem, editor, {\em 20th Annual Conference on Computer
  Science Logic (CSL `11), September 12-15, 2011, Bergen, Norway, Proceedings},
  volume~12 of {\em LIPIcs}, pages 443--457. Schloss Dagstuhl - Leibniz-Zentrum
  fuer Informatik, 2011.

\bibitem{n51am}
John Nash.
\newblock Non-cooperative games.
\newblock {\em The Annals of Mathematics}, 54(2):pp. 286--295, 1951.

\bibitem{nn11lori}
Pavel Naumov and Brittany Nicholls.
\newblock Game semantics for the {G}eiger-{P}az-{P}earl axioms of independence.
\newblock In {\em The Third International Workshop on Logic, Rationality and
  Interaction (LORI-III), LNAI 6953}, pages 220--232. Springer, 2011.

\bibitem{ns12loft}
Pavel Naumov and Italo Simonelli.
\newblock Strict equilibria interchangeability in multi-player zero-sum games.
\newblock In {\em 10th Conference on Logic and the Foundations of Game and
  Decision Theory (LOFT)}, 2012.

\bibitem{s71jet}
Robert~M Solow and William~S Vickrey.
\newblock Land use in a long narrow city.
\newblock {\em Journal of Economic Theory}, 3(4):430 -- 447, 1971.

\bibitem{s86}
David Sutherland.
\newblock A model of information.
\newblock In {\em Proceedings of Ninth National Computer Security Conference},
  pages 175--183, 1986.

\end{thebibliography}

\end{document}